\documentclass[a4paper, 10 pt, conference]{ieeeconf}  % Comment this line out
\usepackage{amssymb,mathrsfs,amsfonts,dsfont,stmaryrd,marvosym}
\usepackage{mathtools}
\usepackage{color}
\usepackage{graphicx}  
\usepackage{psfrag, xcolor}
\usepackage{subfigure,cite,epsfig,url}

\allowdisplaybreaks[4]
\newtheorem{theorem}{Theorem}

\newtheorem{proposition}{Proposition}   
\newtheorem{definition}{Definition} 
\newtheorem{remark}{Remark} 
 
\newcommand{\mc}{\mathcal} % determenistic vector
\newcommand{\mkv}{-\!\!\!\!\minuso\!\!\!\!-}
\newcommand{\mw}[1]{{\color{black}#1}}
\newcommand{\todo}[1]{{\color{black}#1}}
\newcommand{\pe}[1]{{\color{black}#1}}
\newcommand{\MU}{\lfloor 2^{nR_1} \rfloor}
\newcommand{\MV}{\lfloor 2^{nR_2} \rfloor}
\newcommand{\M}{\mathsf{M}}
\newcommand{\T}{\mathsf{T}}
\newcommand{\W}{\mathsf{W}}
\renewcommand{\L}{\mathsf{L}}

\graphicspath{{figs/}}
\newcommand{\tp}[1]{\textnormal{tp}(#1)}

\begin{document}
	\fontencoding{OT1}\fontsize{10}{11}\selectfont
	
	\title{Distributed Hypothesis Testing with Collaborative Detection}
	\author{Pierre Escamilla$^{\dagger}$$^{\star}$ \qquad Abdellatif Zaidi$^{\dagger}$ $^{\ddagger}$ \qquad Mich\`ele Wigger $^{\star}$ \vspace{0.3cm}\\
		$^{\dagger}$ Paris Research Center, Huawei Technologies, Boulogne-Billancourt, 92100, France\\
		$^{\ddagger}$ Universit\'e Paris-Est, Champs-sur-Marne, 77454, France\\
		$^{\star}$ LTCI, T\'el\'ecom ParisTech, Universit\'e Paris-Saclay, 75013 Paris, France\\
		\{\tt pierre.escamilla@huawei.com, abdellatif.zaidi@u-pem.fr\}\\
		\{\tt michele.wigger@telecom-paristech.fr\}
		%\vspace{-5mm}
	}
	
	\maketitle

	\begin{abstract}
		A  detection system with a single sensor and two detectors is considered, where each of the terminals observes a memoryless source sequence, the  sensor sends a message  to both  detectors and the first detector  sends a message to the  second detector. Communication of these messages is assumed to be error-free but rate-limited. The joint probability mass function (pmf) of the  source sequences observed at the three terminals depends on an $\M$-ary hypothesis $(\M \geq 2)$, and the goal of the communication is that each detector can guess the underlying hypothesis. Detector $k$, $k=1,2$, aims to maximize the error exponent \textit{under hypothesis} $i_k$, $i_k \in \{1,\ldots,\M\}$, while ensuring a small probability of error under all other hypotheses. We study this problem in the case in which the detectors \mw{aim to maximize} their error exponents under the \textit{same} hypothesis (i.e., $i_1=i_2$) and in the case in which they \mw{aim to maximize} their error exponents under \textit{distinct} hypotheses (i.e., $i_1 \neq i_2$). For the setting in which $i_1=i_2$, we present an achievable exponents region for the case of positive communication rates, \mw{and show that it is optimal for a specific case of testing against independence. We also} characterize the optimal exponents region in the case of zero communication rates. For the setting in which $i_1 \neq i_2$, we characterize the \mw{optimal exponents region} in the case of zero communication rates. % \mw{when sufficiently many bits can be sent}. %\pe{ The analysis in this case reveals that a threshold on the number of communication bits exists below which a dilemma between the error exponents appears.} 
		
		% A set of achievable Type-II error exponent pairs  proposed, under the constraint that the Type-I error probabilities at both receivers are bounded by  small values $\epsilon_1$ and $\epsilon_2$. The  proposed set of achievable Type-II error exponents is optimal in special cases. For example when both rates of communication  are 0, or when only the rate of communication from the first detector to the second is 0, the detectors have independent observations and one is  testing against independence.

		%Using the criterion that stipulates minimization of the Type II error probability subject to a constant-type constraint on the Type I error probability, we derive an inner bound on the region of error exponents achieved by such a decision system. Furthermore, we show through a strong converse proof that these error exponents are optimal if \textit{zero-rate} communication is assumed. We also specialize the results to some specific cases of testing against independence.
	\end{abstract}

	\section{Introduction}~\label{secI}
	Consider the multiterminal hypothesis testing scenario shown in Figure~\ref{fig-system-model}, where an encoder observes a discrete memoryless source sequence $X^n\triangleq (X_1,\ldots, X_n)$  and communicates with two remote detectors 1 and 2 over a \textit{common} noise-free bit-pipe of rate $R_1\geq 0$. Here, $n$ is a positive integer that denotes the blocklength. Detectors  1 and 2 observe correlated memoryless source sequences $Y^n_1\triangleq(Y_{1,1},\ldots ,Y_{1,n})$ and $Y^n_2\triangleq (Y_{2,1},\ldots ,Y_{2,n})$, and Detector 1 can communicate with Detector 2 over a noise-free bit-pipe of rate $R_2$. The sequence of observation triples $\{(X_{t}, Y_{1,t}, Y_{2,t})\}_{t=1}^n$ is independent and identically distributed (i.i.d) according to a joint probability mass function (pmf) that is determined by the hypothesis $\mc H \in \{1,\ldots,\M\}$. Under the hypothesis $\mc H=m$,
	\begin{equation}\label{eq:H}
	\{(X_{t}, Y_{1,t}, Y_{2,t})\}_{t=1}^n \textnormal{  i.i.d. } \sim P^{(m)}_{XY_1Y_2}.
	\end{equation}
	
	Detector 1 decides on a hypothesis $\hat{\mc{H}}_1 \in \{1,\ldots,\M\}$ with the goal to maximize the exponential decrease of the probability of \pe{ error \textit{under hypothesis} $\mc H = i_1 \in \{1,\ldots,\M\}$ (i.e., guessing $\hat{\mc H}_1 \neq i_1$ when $\mc H = i_1$),} while ensuring that the probability of \pe{error under $\mc H = m$ with $m \neq i_1$ (i.e., guessing $\hat{\mc H}_1 \neq  m$)} does not exceed some prescribed constant value $\epsilon_1 \in (0,1)$ for all sufficiently large blocklengths $n$. Similarly, Detector 2 decides on a hypothesis $\hat{\mc H}_2 \in \{1,\ldots,\M\}$ with the goal to maximize the exponential decrease of the probability of \pe{ error \textit{under hypothesis} $\mc H = i_2 \in \{1,\ldots,\M\}$ (i.e., guessing $\hat{\mc H}_2 \neq i_2$ when $\mc H = i_2$)}, while ensuring that the probability of \pe{error under $\mc H = m$ with $m \neq i_2$ (i.e., guessing $\hat{\mc H}_2 \neq m$)} does not exceed a constant value $\epsilon_2 \in (0,1)$ for all sufficiently large blocklengths $n$.

	In this paper, we study the problem of how cooperation among the two detectors can be used to improve the largest error exponents. We investigate this problem in both settings, the one in which the detectors aim at maximizing the error exponents under the \textit{same} hypothesis (i.e., $i_1=i_2$) and the one in which they aim at maximizing the error exponents under \textit{different} hypotheses (i.e., $i_1 \neq i_2$). 
	
	\begin{figure}[!t]
		\begin{center}
			\includegraphics[width=\linewidth]{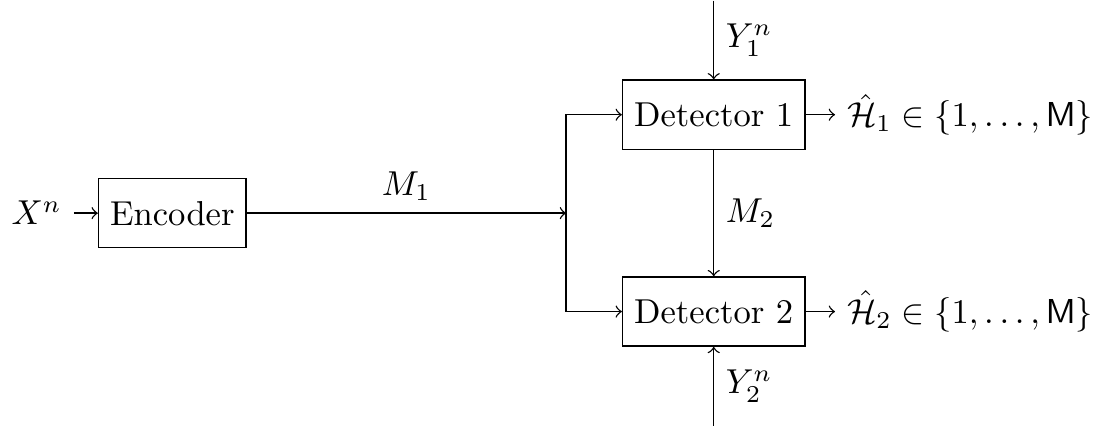}
			\caption{A Heegard-Berger type source coding model with unidirectional conferencing for multiterminal hypothesis testing.}
			\label{fig-system-model}
		\end{center}
		\vspace{-5mm}
	\end{figure}
	
	\subsection{Related Works}
	
\mw{Problems of distributed hypothesis testing are strongly rooted in both statistics and information theory. In particular,  the problem described above but without Detector 2 was studied in \cite{AC86, han_hypothesis_1987}. In \cite{AC86},   Ahlswede and Csisz\'ar  presented a single-letter lower bound on the largest possible error exponent.  It is optimal in the  special case of \emph{testing against independence}, but for the general case was  improved by Han in \cite{han_hypothesis_1987}.  Extensions of these results to networks with multiple encoders, multiple detectors, or interaction between terminals, can be found, e.g., in~\cite{wigger2016testing,SWT18,isit2018,xiang2012interactive,katz2016collaborative,zhao2014distributed,ZL15,SHA94,rahman2012optimality}. In particular, \cite{isit2018} studies the model considered here but without cooperation. \cite{wigger2016testing} and \cite{SWT18} study the model of\cite{isit2018} in the specific case of testing against independence and conditional independence, respectively.

	%In particular, \cite{wigger2016testing,SWT18,isit2018} study one of the models considered here but without communication between detectors. 

	Han \cite{han_hypothesis_1987}  also  introduced a two-terminal  hypothesis testing problem where the  encoder sends a single bit to the decoder.  He established a single-letter characterization of the optimal error exponent of this system. 	
	Subsequently, Shalaby and Papamarcou~\cite{shalaby1992multiterminal} showed that Han's error exponent remains optimal even when the transmitter can send a sublinear number of bits to the receiver.  Key to the derivation of the converse proof in~\cite{han_hypothesis_1987} is an ingenious use of the ``Blowing-Up" lemma~\cite[Theorem 5.4]{csiszar2011information}. This lemma plays a similar crucial role for establishing converse parts for more general zero-rate hypothesis testing systems with exponential-type constraint on \pe{ all errors}\cite{HK89}. }

	\subsection{Focus and Main Contributions}
	
	As we already mentioned, one major goal in this paper is the study of the role of cooperation link between the two detectors in improving the error exponents, i.e., \textit{collaborative} decision making. On this aspect, we mention that the presence of the cooperation link makes the problem depart significantly from the aforementioned works. To see this, observe for example that even the seemingly easy case in which i) the rate $R_1$ is zero and ii) $M=2$ and both detectors making guesses on the same hypothesis about whether $(Y_1,Y_2)$ is independent of $X$ or not, which is solved fully in~\cite{wigger2016testing} in the case without cooperation, seems to become of formidable complexity in the presence of such a cooperation link. Partly, this is because \textit{binning} on the cooperation link may now be helpful in this scenario (See Remark~\ref{remak-test-against-independence}).

	 Another goal in this paper is to investigate the effect of the detectors aiming at maximizing their  error exponents under \textit{distinct} hypotheses. \mw{Such a scenario was already studied in \cite{isit2018}, but here we investigate it in a collaborative setup.}

\mw{The contributions of this paper are as follows. For the case} in which the detectors aim \mw{to maximize} their error exponents under the same hypothesis (i.e., $i_1=i_2$), we propose a coding and testing scheme for positive rates $R_1 \geq 0$ and $R_2 \geq 0$. Based on this scheme, we present an achievable exponents region for the case of positive communication rate, and we characterize the optimal exponents region for the case of zero communication rate. We also specialize the results to some specific cases of testing against independence for which we find the optimal error exponents. For the setting in which the detectors aim \mw{to maximize} their error exponents under distinct hypotheses, we characterize the optimal exponents region for the case of zero communication rates. % \pe{ when sufficiently large number of bits}. %Interestingly, the analysis in this case reveals that the error exponents at the detectors exhibit some tension between them if the ranges of the encoding functions are small. The threshold on the number of communication bits is shown to depend on the number of hypotheses with distinct $X$-marginals and the size of the largest set of null hypotheses with identical $X$-marginals.  

	\subsection{Outline and Notation}
	
	The reminder of this paper is organzied as follows. Section~\ref{secII} contains a description of the system model. In Section~\ref{secIII} we study the model in which the two detectors aim at maximizing the error exponents under the same hypothesis, and in Section~\ref{secIV} we study the model in which the two detectors aim at maximizing the error exponents under the different hypotheses. Throughout, we use the following notations. The set of all possible types of $n$-length sequences over $\mathcal{X}$ is denoted  $\mathcal{P}^n(\mathcal{X})$.
	For $\delta>0$, the set of sequences $x^n$ that are $\delta$-typical with respect to the pmf $P_X$ is denoted  $\mc T^n_{\delta}(P_X)$. For random variables $X$ and $\bar{X}$ over the same alphabet $\mathcal{X}$  with pmfs $P_X$ and $P_{\bar{X}}$  satisfying $P_X \ll P_{\bar{X}}$ (i.e., for every $x \in \mathcal{X}$, if $P_{X}(x) > 0 $ then also $P_{\bar{X}}(x) > 0$), both $D(P_X\|P_{\bar{X}})$ and $D(X\|\bar{X})$ denote the Kullback-Leiber divergence between $X$ and $\bar{X}$.
	%\begin{equation*}
	%D(X\|\bar{X}) = D(P_X\|P_{\bar{X}})  = \sum_{x \in \mathcal{X}} P_X(x) \log{\frac{P_X(x)}{P_{\bar{X}}(x)}}
	%\end{equation*}
	%Also, particularly important in this work is the notion of type~\cite{csiszar2011information}. We use notation that is compliant with~\cite{csiszar2011information}. 
	
	\vspace{3cm}
	
	\section{System Model}~\label{secII}
	%Let $\mc X$, $\mc Y_1$ and $\mc Y_2$ be arbitrary finite sets and consider correlated random variables $XY_1Y_2=(X,Y_1,Y_2)$ with joint probability mass function (PMF) $Q_{XY_1Y_2}(x,y_1,y_2)$ for all $(x,y_1,y_2) \in \mc X{\times}\mc Y_1{\times}\mc Y_2$. It is assumed that this joint PMF is one of two possibles PMFs, depending on two hypotheses, a null hypothesis $\mc H$ and an alternative hypothesis $\bar{\mc H}$. Specifically, under the null hypothesis 
	%\begin{equation*}
	%\mc H\: : \: Q_{XY_1Y_2}(x,y_1,y_2) = P_{XY_1Y_2}(x,y_1,y_2)
	%\end{equation*}
	%for all $(x,y_1,y_2) \in \mc X{\times}\mc Y_1{\times}\mc Y_2$; and, under the alternative hypothesis 
	%\begin{equation*}
	%\bar{\mc H}\: : \: Q_{XY_1Y_2}(x,y_1,y_2) = P_{\bar{X}\bar{Y}_1\bar{Y}_2}(x,y_1,y_2) 
	%\end{equation*}
	%for all $(x,y_1,y_2) \in \mc X{\times}\mc Y_1{\times}\mc Y_2$, where $P_{XY_1Y_2}$ and $P_{\bar{X}\bar{Y}_1\bar{Y}_2}$ are two joint probability distributions on $\mc X{\times}\mc Y_1{\times}\mc Y_2$. In accordance with the related literature and the information theoretic terminology, we will sometimes write simply 
	%\begin{align*}
	%\mc H \: & : \: XY_1Y_2 \\
	%\bar{\mc H} \: &: \: \bar{X}\bar{Y}_1\bar{Y}_2.
	%\end{align*}
	
	Let $(X^n,Y^n_1,Y^n_2)$ be distributed i.i.d. according to one of $M \geq 2$ possible pmfs $\{P^{(m)}_{XY_1Y_2}\}_{m=1}^M$.   
 The encoder observes a source sequence $X^n$ and  applies encoding function
	\begin{equation}\label{eq:phi1}
	\phi_{1,n}\colon  \mc X^n \rightarrow \mc M_1\triangleq\{1,\ldots,\|\phi_{1,n}\|\}
	\end{equation} 
	to it. It then sends the resulting index 
	\begin{equation}
	M_1 = \phi_{1,n}(X^n)
	\end{equation}
	to both decoders. 
	
	Besides $M_1$, Decoder 1 also observes the source sequence $Y_1^n$. It applies two functions to the pair $(M_1, Y_1^n)$: an encoding function 
	\begin{equation}\label{eq:phi2}
	\phi_{2,n} \colon \mc M_1 \times \mc Y^n_1 \rightarrow \mc M_2 \triangleq\{1,\ldots,\|\phi_{2,n}\|\},
	\end{equation}
	and a decision function 
	\begin{equation}\label{eq:psi1}
	\psi_{1,n} \colon \mc M_1 \times \mc Y^n_1 \rightarrow \{1,\ldots,\M\}. % \{\mc H, \bar{\mc H}\}. 
	\end{equation}
	It sends the index 
	\begin{equation}\label{eq:M2}
	M_2=\phi_{2,n} (M_1, Y_1^n)
	\end{equation} to Decoder 2, and decides on the hypothesis 
	\begin{equation}\hat{\mc H}_1 \triangleq \psi_{1,n} (M_1, Y_1^n).
	\end{equation}

	\noindent Decoder 2 observes  $(M_1, M_2, Y_2^n)$ and applies the decision function 
	\begin{equation}\label{eq:psi2}
	\psi_{2,n}\colon  \mc M_1 \times \mc M_2 \times \mc Y^n_2 \rightarrow \{1,\ldots,\M\} % \{\mc H, \bar{\mc H}\}
	\end{equation}
	to this triple. It thus decides on the hypothesis  
	\begin{equation}\hat{\mc H}_2 \triangleq  \psi_{2,n} (M_1, M_2, Y_2^n).
	\end{equation} 
	
	%The communication from the encoder and Decoder 1 have rates $R_1$ and $R_2$, respectively, in the sense that for any fixed $\eta > 0$ and sufficiently large $n$ (depending on $\eta$):
	%\begin{subequations}
	%\begin{align}
	%\frac{1}{n}\log \|\phi_{1,n}\| &\leq R_1 + \eta \\
	%\frac{1}{n}\log \|\phi_{2,n}\| &\leq R_2 + \eta.
	%\end{align}
	%\label{eq-definition-achievable-rates}
	%\end{subequations}
	%
	%and the decoding function at Decoder 2 is a mapping
	%
	
	%\noindent Let $\mathscr{A}_{1,n}$ denote the \textit{acceptance region} for hypothesis $\mc H$ at Decoder 1, i.e.,
	%\begin{equation}
	%\mathscr{A}_{1,n} \triangleq \left\{ \left(x^n,y_1^n\right) \in \mathcal{X}^n \times \mathcal{Y}_1^n \: : \: \psi_{1,n}\left(\phi_{1,n}(x^n),y_1^n\right) = \mc H  \right\}.
	%\end{equation}
	%Similarly, let $\mathscr{A}_{2,n}$ denote the \textit{acceptance region} for hypothesis $\mc H$ at Decoder 2, i.e.,
	%
	%\begin{align}
	%	\mathscr{A}_{2,n} \triangleq & \big\{ \left(x^n,y_1^n,y_2^n\right) \in \mathcal{X}^n \times \mathcal{Y}_1^n \times \mathcal{Y}_2^n \: : \nonumber\\ 
	%		& \psi_{2,n}\left(\phi_{1,n}(x^n),\phi_{2,n}\left(\phi_{1,n}(x^n),y_1^n\right),y_2^n\right) = \mc H \big\}.
	%\end{align}
	%The \textit{critical region} $\mathscr{A}^c_{1,n}$ is the complement of $\mathscr{A}_{1,n}$ in $\mc X^n{\times}\mc Y^n_1$ and the \textit{critical region} $\mathscr{A}^c_{2,n}$ is the complement of $\mathscr{A}_{2,n}$ in $\mc X^n{\times}\mc Y^n_1{\times}\mc Y^n_2$.
	%
	
	\noindent Let $(i_1,i_2) \in \{1,\ldots,\M\}^2$ be given. The probabilities at Decoder~1 and Decoder~2 are given by
	\begin{IEEEeqnarray}{rCl}
	\alpha_{1,n} &\triangleq &\max_{m \neq i_1}\text{Pr}\big\{\hat{\mc H}_1 \neq m \big| \mc H = m \big\},\:\\ 
	\beta_{1,n}  & \triangleq& \text{Pr}\big\{\hat{\mc H}_1 \neq i_1 \big| \mc H = i_1 \big\} \nonumber\\
	\alpha_{2,n} & \triangleq& \max_{m \neq i_2}  \text{Pr}\big\{\hat{\mc H}_2 \neq m \big|\mc H = m \big\},\:\\  \beta_{2,n}  & \triangleq& \text{Pr}\big\{\hat{\mc H}_2 \neq i_2 \big|\mc H = i_2  \big\}. \nonumber
	\end{IEEEeqnarray}
	
	%In this paper, we are interested in determining the minimum probabilities $(\beta_{1,n},\beta_{2,n})$ of the Type II errors, under the condition that the Type I errors satisfy for all sufficiently large blocklengths $n$:
	%\begin{subequations}
	%\begin{align}
	%\alpha_{1,n} &\leq \epsilon_1 \\
	%\alpha_{2,n} &\leq \epsilon_2
	%\end{align}
	%\label{eq-definition-constant-constraints-typeI-errors}
	%\end{subequations}
	%for some prescribed $0 < \epsilon_1 <1$ and $0 < \epsilon_2 <1$. Formal definitions of the \textit{power-exponents} are as follows.

	\begin{definition}%[Achievable Type-II Error Exponents under Constant Constraints on Type I Errors] 
		Given rates $R_1, R_2 \geq 0$ and small positive numbers $\epsilon_1, \epsilon_2 \in (0,1)$, an \mw{error-exponents} pair $(\theta_1,\theta_2)$ is said achievable, if for each blocklength $n$ there exist  functions $\phi_{1,n}$, $\phi_{2,n}$, $\psi_{1,n}$ and $\psi_{2,n}$ as in \eqref{eq:phi1}, \eqref{eq:phi2}, \eqref{eq:psi1}, and \eqref{eq:psi2} so that the following limits hold:
		\begin{align}
		\varlimsup_{n\to \infty} \alpha_{1,n} \leq \epsilon_1, \quad \varlimsup_{n\to \infty}\alpha_{2,n} \leq \epsilon_2,
		\label{eq-definition-constant-constraints-typeI-errors}
		\end{align}
		\begin{IEEEeqnarray}{rCl}
		\theta_1 &\leq   \varliminf_{n \to \infty}-\frac{1}{n} \log \beta_{1,n}, \quad 
		\theta_2& \leq  \varliminf_{n \to \infty}-\frac{1}{n} \log \beta_{2,n},
		\end{IEEEeqnarray}
		and 
		\begin{align}
		\varlimsup_{n\to \infty} \frac{1}{n}\log \|\phi_{1,n}\| &\leq R_1, \quad
		\varlimsup_{n\to \infty} \frac{1}{n}\log \|\phi_{2,n}\| &\leq R_2 .
		\label{eq-definition-achievable-rates}
		\end{align}
		%define $\Epsilon(R_1, R_2, \epsilon_1, \epsilon_2)$ as the set of all error exponents $(\theta_1, \theta_2)$ 
		%\begin{align}
		%\beta_{1,n}(R_1,R_2,\epsilon_1,\epsilon_2,\eta) = \min_{\phi_{1,n},\psi_{1,n}} \beta_{1,n} \\
		%\beta_{2,n}(R_1,R_2,\epsilon_1,\epsilon_2,\eta) = \min_{\phi_{1,n},\phi_{2,n},\psi_{2,n}} \beta_{2,n}
		%\end{align}
		%where $\phi_{1,n}$, $\phi_{2,n}$, $\psi_{1,n}$ and $\psi_{2,n}$ range over all possible functions satisfying~\eqref{eq:phi1}--\eqref{eq:psi2}  and~\eqref{eq-definition-constant-constraints-typeI-errors}. Furthermore, define
		%\begin{align}
		%\theta_1(R_1,R_2,\epsilon_1,\epsilon_2,\eta) = \liminf_{n \to \infty} \left(-\frac{1}{n} \log \beta_{1,n}(R_1,R_2,\epsilon_1,\epsilon_2,\eta) \right) \\
		%\theta_2(R_1,R_2,\epsilon_1,\epsilon_2,\eta) = \liminf_{n \to \infty} \left(-\frac{1}{n} \log \beta_{2,n}(R_1,R_2,\epsilon_1,\epsilon_2,\eta) \right).
		%\end{align}
		%The \textit{error  exponents} $\theta_1(R_1,R_2,\epsilon_1,\epsilon_2)$ and $\theta_2(R_1,R_2,\epsilon_1,\epsilon_2)$ for the hypothesis testing $\mc H$ against $\bar{\mc H}$ are defined as
		%\begin{align}
		%\theta_1(R_1,R_2,\epsilon_1,\epsilon_2) &= \lim_{\eta \to 0} \theta_1(R_1,R_2,\epsilon_1,\epsilon_2,\eta) \\
		%\theta_2(R_1,R_2,\epsilon_1,\epsilon_2) &= \lim_{\eta \to 0} \theta_2(R_1,R_2,\epsilon_1,\epsilon_2,\eta).
		%\end{align}
		%\qed
	\end{definition}
	% In this work, we develop an inner bound on the region of feasible $\theta_1(R_1,R_2,\epsilon_1,\epsilon_2)$ and $\theta_2(R_1,R_2,\epsilon_1,\epsilon_2)$ for general hypotheses $\mc H$ and $\bar{\mc H}$ and arbitrary positive rates $R_1$ and $R_2$. Furthermore, we establish a single-letter characterization of $\theta_1(0,0,\epsilon_1,\epsilon_2)$ and $\theta_2(0,0,\epsilon_1,\epsilon_2)$ in the case of \textit{zero-rate}  hypothesis testing, as well as some specific cases of testing against independence with positive rates.
	
	\begin{definition}
		Given rates $R_1, R_2 \geq 0$ and  numbers $\epsilon_1, \epsilon_2 \in (0,1)$, the  closure of the set of all achievable exponent pairs $(\theta_1, \theta_2)$ is called the \emph{\mw{error-exponents} region $\mathcal{E}(R_1, R_2, \epsilon_1, \epsilon_2)$.} % For rate $R_1=0$, we will sometimes make a finer distinction by referring to the case in which $\|\phi_{1,n}\| \leq \W_1$ as $R_1=0_{\mathsf{W}_1}$. A similar notation will be used for when $\|\phi_{2,n}\| \leq \W_2$.
	\end{definition}
	
	The main interest  of this paper is on characterizing the set of all achievable \mw{error-exponent} pairs. To do so, we distinguish the case in which the two detectors aim at maximizing the error exponents under the \textit{same} hypothesis, i.e., $i_1=i_2$, and the one in which they aim at maximizing the error exponents under \textit{different} hypotheses, i.e., $i_1 \neq i_2$. 
	
	\section{Cooperative Detection}~\label{secIII}
	
	In this section, we study the setting in which the two detectors aim at maximizing the error exponents under the \textit{same} hypothesis, i.e., $i_1=i_2$. For simplicity, we first consider the case of simple null hypothesis, i.e., $M=2$. For simplicity we set $i_1=i_2=2$, and   replace $P^{(1)}$ by  $P$ and $P^{(2)}$ by $\bar{P}$. For convenience, we assume that $\bar{P} (x, y_1,y_2)>0$ for all $(x,y_1,y_2)\in \mathcal{X} \times \mathcal{Y}_1\times \mathcal{Y}_2$.
	
	% In this part we focus on the case where there are only two possible hypothesis ($k \in \{\mc H, \bar{ \mc H}\}$) and that $\bar{\mc H}_1= \bar{ \mc{H}}_2 = \bar{ \mc H} $. We will note  respectively $P_{XY_1Y_2}$ and $P_{\bar X \bar Y_1 \bar Y_2}$ the pmfs under $\mc H$ and $\bar{ \mc H}$.
	
	%\subsection{Positive Communication Rates}

	Our first result is an inner bound on the error-exponents region  $\mathcal{E}(R_1, R_2, \epsilon_1, \epsilon_2)$.
	To state the results, we make the following definitions.
	For given rates $R_1 \geq 0$ and $R_2 \geq 0$, define the following set of auxiliary random variables:
	\[ \mathcal{S} \left(R_1,R_2\right) \triangleq \left\{ \left(U,V\right) \colon \left. \begin{array}{c} U \mkv X \mkv (Y_1, Y_2 )\\ V \mkv( Y_1,U )\mkv (Y_2,X) \\ I\left(U;X\right) \leq R_1 \\ I\left(V;Y_1|U\right) \leq R_2 \end{array} \right.\right\}.\]
	Further, define for each $(U,V) \in \mathcal{S} \left(R_1,R_2\right)$:
	\begin{align*}
	\mathcal{L}_1\left(U\right) \triangleq \left\{(\tilde{U},\tilde{X},\tilde{Y}_1) \colon P_{\tilde{U}\tilde{X}} = P_{UX},\  P_{\tilde{U}\tilde{Y}_1} = P_{UY_1}\right\},
	\end{align*}
	and
	\begin{align*}
	\mathcal{L}_2\left(UV\right) &\triangleq \left\{(\tilde{U},\tilde{V},\tilde{X},\tilde{Y}_1,\tilde{Y}_2) \colon  \begin{array}{c}   P_{\tilde{U}\tilde{X}} = P_{UX}\\  P_{\tilde{U}\tilde{V}\tilde{Y}_1} = P_{UVY_1}\\  P_{\tilde{U}\tilde{V}\tilde{Y}_2} = P_{UVY_2} \end{array} \right\}. 
	\end{align*}
	Also, let  $(\bar{X}, \bar{Y}_1, \bar{Y}_2) \sim P_{\bar{X} \bar{Y}_1\bar{Y}_2}$ and 
		define the random variables $(\bar{U}, \bar{V})$ so as to satisfy $P_{\bar{U}|\bar{X}}=P_{U|X}$ and $P_{\bar{V}|\bar{Y_1}\bar{U}}=P_{V|Y_1U}$ and the Markov chains
	\begin{IEEEeqnarray}{rCl}
	\bar{U} \mkv \bar{X} \mkv (\bar{Y}_1,\bar{Y}_2) \\
	\bar{V} \mkv (\bar{Y_1},\bar{U}) \mkv (\bar{X},\bar{Y}_2).
	\end{IEEEeqnarray}

	\begin{theorem}[Positive Rates]~\label{theorem-lower-bounds-power-exponents-general-hypotheses-positive-rates} 
		Given rates $R_1 \geq 0$ and $R_2 \geq 0$ and numbers $\epsilon_1,  \epsilon_2 \in(0,1)$, the exponents region $\mathcal{E}(R_1, R_2,\epsilon_1, \epsilon_2)$ contains all nonnegative  pairs $(\theta_1,\theta_2)$ that for some $(U,V) \in \mathcal{S} \left(R_1,R_2\right)$ satisfy:
		\begin{align}\label{eq: an achievable rate exponent region for Heegard-Berger}
		& \theta_1 \leq  \min_{\tilde{U}\tilde{X}\tilde{Y}_1 \in \mathcal{L}_1\left(U\right)} D\left(\tilde{U}\tilde{X}\tilde{Y}_1||\bar{U}\bar{X}\bar{Y}_1\right)   \\
		& \theta_2 \leq \min_{\tilde{U}\tilde{V}\tilde{X}\tilde{Y}_1\tilde{Y}_2 \in \mathcal{L}_2\left(UV\right)} D\left(\tilde{V}\tilde{U}\tilde{X}\tilde{Y}_1\tilde{Y}_2||\bar{V}\bar{U}\bar{X}\bar{Y}_1\bar{Y}_2\right).
		\end{align}
		
	\end{theorem}
	
	\vspace{0.2cm}
	
	\begin{proof}
		See Section~\ref{secV_subsecA}. 
	\end{proof}
	
	Theorem~\ref{theorem-lower-bounds-power-exponents-general-hypotheses-positive-rates} characterizes an inner bound on the exponent rate region $\mathcal{E}(R_1, R_2, \epsilon_1, \epsilon_2)$. The following results indicate that in some scenarios the inner bound coincides  with $\mathcal{E}(R_1, R_2, \epsilon_1, \epsilon_2)$.
	\begin{proposition} If
		\begin{equation}\label{eq:rates}
		R_1 \geq H(X)\quad \textnormal{and} \quad R_2 \geq H(Y_1|X),
		\end{equation}
		for any $\epsilon_1,\epsilon_2 \in(0,1)$ the exponents region $\mathcal{E}(R_1, R_2, \epsilon_1, \epsilon_2)$ coincides with the set of all  non-negative  pairs $(\theta_1, \theta_2)$ satisfying
		\begin{align}
		\theta_1 &\leq  D(XY_1\|\bar{X}\bar{Y}_1)\label{eq:cons1} \\
		\theta_2 &\leq D(XY_1Y_2\|\bar{X}\bar{Y}_1\bar{Y}_2).\label{eq:cons2}
		\end{align}
	\end{proposition}
	
	\begin{proof} Achievability follows by specializing Theorem~\ref{theorem-lower-bounds-power-exponents-general-hypotheses-positive-rates} to  $U \triangleq X$ and $V\triangleq Y_1$ and by noting that this choice is feasible, i.e., in $ \mathcal{S} \left(R_1,R_2\right)$ because of \eqref{eq:rates}. The converse holds because the right-hand side of \eqref{eq:cons1} coincides with the error exponent when Decoder~1 can directly observe the source sequences $X^n$ and $Y_1^n$, and the right-hand side of \eqref{eq:cons2} coincides with the error exponent when Decoder~2 can directly observe the source sequences $X^n$, $Y_1^n$, and $Y_2^n$. 
		%Set $(U,V) \in \mc S(R_1,R_2)$ to be $U:=X$ and $V:=Y_1$, which satisfy $I(U;X) \leq R_1$ and $I(V;Y_1|U) \leq R_2$. Then, $(\tilde{U}\tilde{X}\tilde{Y}_1) \in \mc L_1(U)$ implies that $P_{\tilde{U}\tilde{X}\tilde{Y}_1}=P_{UXY_1}$; and, hence, $D(\tilde{U}\tilde{X}\tilde{Y}_1\|\bar{U}\bar{X}\bar{Y}_1)=D(UXY_1\|\bar{U}\bar{X}\bar{Y}_1)=D(XY_1\|\bar{X}\bar{Y}_1)$. Similarly, $D(\tilde{U}\tilde{V}\tilde{X}\tilde{Y}_1\tilde{Y}_2\|\bar{U}\bar{V}\bar{X}\bar{Y}_1\bar{Y}_2)=D(UVXY_1\|\bar{U}\bar{V}\bar{X}\bar{Y}_1\bar{Y}_2)=D(XY_1Y_2\|\bar{X}\bar{Y}_1\bar{Y}_2)$. The inverse inequalities are obvious.
	\end{proof}

	Consider now the case of zero cooperation rate $R_2=0$. For given rate $R_1\geq 0$, define 
	\begin{equation}
	\mc S(R_1) \triangleq \big\{U \colon  I(U;X) \leq R_1, \; U \mkv X \mkv (Y_1,Y_2)\big\}.
	\end{equation}
	
	\begin{theorem}[Zero Cooperation Rate]~\label{theorem-lower-bounds-power-exponents-test-against-independence-positive-rates}
		%If the cooperation rate 
		%\begin{equation}\label{eq:rate3}
		%R_2=0,
		%\end{equation} 
		If the pmfs $P_{XY_1Y_2}$ and $P_{\bar{X}\bar{Y}_1\bar{Y}_2}$ satisfy
		\begin{IEEEeqnarray}{rCl}
		P_{Y_1Y_2} &= P_{Y_1} P_{Y_2} , \quad P_{\bar{X}\bar{Y}_1\bar{Y}_2} &= P_XP_{Y_1}P_{Y_2},
		\end{IEEEeqnarray}
		then the asymptotic region
		\begin{equation}\bigcap_{\epsilon_1,\epsilon_2>0}{\mathcal{E}(R_1, 0, \epsilon_1, \epsilon_2)}
		\end{equation} coincides with the set of all nonnegative pairs $(\theta_1, \theta_2)$ that for  some $U \in \mc S(R_1)$ satisfy
		%P_{\bar{X}}=P_X$, $P_{\bar{Y}_1\bar{Y}_2}=P_{Y_1}P_{Y_2}$ and $P_{\bar{X}\bar{Y}_1\bar{Y}_2}=P_XP_{Y_1}P_{Y_2}$. For given rate $R_1 \geq 0$, and $0 < \epsilon_1 <1$ and $0 < \epsilon_2 <1$, 
		%the error exponents  $\theta_1(R_1,0,\epsilon_1,\epsilon_2)$ and $\theta_2(R_1,0,\epsilon_1,\epsilon_2)$ are given by
		\begin{align}
		\label{power-exponent-test-vs-indep-with-indep-si-Decoder1}
		\theta_1 & \leq  I\left(U;Y_1\right) \\
		\theta_2 & \leq  I\left(U;Y_1\right) + I\left(U;Y_2\right).
		\label{power-exponent-test-vs-indep-with-indep-si-Decoder2}
		\end{align}
	\end{theorem}
	% \vspace{0.2cm}
	\begin{proof} Achievability follows by specializing Theorem~\ref{theorem-lower-bounds-power-exponents-general-hypotheses-positive-rates}  to $R_2=0$. The form in \eqref{power-exponent-test-vs-indep-with-indep-si-Decoder1} and \eqref{power-exponent-test-vs-indep-with-indep-si-Decoder2} is then obtained through algebraic manipulations and by using the log-sum inequality. For the converse, see Section~\ref{secV_subsecB}.
		%case in which $R_2=0$ and $P_{\bar{X}\bar{Y}_1\bar{Y}_2}=P_XP_{Y_1}P_{Y_2}$. The proof of converse of Theorem~\ref{theorem-lower-bounds-power-exponents-test-against-independence-positive-rates} appears in Section~\ref{secV_subsecB}.
	\end{proof}
	
	Notice that the Theorem remains valid if only a single bit can be sent over the cooperation link. In fact it suffices that Detector 1 sends its decision to Detector 2. The latter then declares the null hypothesis $\mathcal{H}=0$, if and only if,   the message from Detector 1 indicates the null hypothesis  and also its own observation combined with the message from the encoder indicate the null hypothesis.
	\begin{remark}~\label{remak-test-against-independence}
		Theorem~\ref{theorem-lower-bounds-power-exponents-test-against-independence-positive-rates} requires that $R_2=0$ and the observations $Y_1$ and $Y_2$ are independent under both null and alternative hypotheses. The reader may wonder whether a similar optimality result can be obtained when these assumptions are relaxed, e.g., both detectors making a guess on whether $(Y_1,Y_2)$ is independent of $X$ or not with $R_2 \geq 0$ and $Y_1$ and $Y_2$ arbitrarily correlated. Such a result however can certainly not be obtained from Theorem~\ref{theorem-lower-bounds-power-exponents-test-against-independence-positive-rates}, because the communication over the cooperation link does employ binning to exploit Detector 2's side-information $Y_2^n$ about the source $Y_1^n$.
		% among them and with $X$ under the alternative hypothesis. Such a result seems elusive, however. To see this, observe that binning on the cooperation link may be useful for improving the error exponent of the second detector in this case.   
	\end{remark}
	
	\begin{remark}
		For the model of Theorem~\ref{theorem-lower-bounds-power-exponents-test-against-independence-positive-rates} without the cooperation link, the optimal error exponent at Detector 2 is $I(U;Y_2)$ only \cite{wigger2016testing}. The $I(U;Y_1)$-increase of this exponent is made possible by the cooperation on the link of rate $R_2$.% The reader may notice that such an improvement is possible even with just one bit, i.e., $R_2= 0_{2}$.
	\end{remark}
	
	%\subsection{Zero Communication Rates}
	
	We now consider the case of zero rates $R_1=R_2=0$. Define the following sets
	\begin{align}
	\label{definition-set1-theorem-optimal-power-exponents-zero-rate-compression}
	\mc L_1 &= \left\{ (\tilde{X},\tilde{Y}_1) \colon  P_{\tilde{X}}=P_X,  P_{\tilde{Y}_1}=P_{Y_1} \right\} \\
	\mc L_2 &= \left\{ (\tilde{X},\tilde{Y}_1,\tilde{Y}_2) \colon P_{\tilde{X}}=P_X,  P_{\tilde{Y}_1}=P_{Y_1}, P_{\tilde{Y}_2}=P_{Y_2} \right\}.
	\label{definition-set2-theorem-optimal-power-exponents-zero-rate-compression}
	\end{align}

	\begin{theorem}[Zero Rates]\label{theorem-optimal-power-exponents-zero-rate-compression}
		If  all $(x,y_1\pe{ ,y_2}) \in \mc X\times\mc Y_1\pe{ \times\mc Y_2}$ have  positive probabilities under $\mathcal{H}=1$,
		%\begin{equation}
		%R_1=R_2=0
		%\end{equation} 
		%and 
		$
		P_{\bar{X}\bar{Y}_1\pe{\bar{Y}_2}}(x,y_1\pe{,y_2}) > 0, 
		$
		then for any $\epsilon_1, \epsilon_2 \in (0,1)$ the exponents region $\mathcal{E}(0, 0, \epsilon_1, \epsilon_2)$ coincides with the set of all nonnegative  pairs $(\theta_1, \theta_2)$ satisfying
		\begin{align}
		\label{power-exponent-zero-rates-Decoder1}
		\theta_1 &\leq  \min_{\tilde{X}\tilde{Y}_1 \in \mathcal{L}_1 } D\left(\tilde{X}\tilde{Y}_1||\bar{X}\bar{Y}_1\right) \\ 
		\theta_2 &\leq  \min_{\tilde{X}\tilde{Y}_1\tilde{Y}_2 \in \mathcal{L}_2} D\left(\tilde{X}\tilde{Y}_1\tilde{Y}_2||\bar{X}\bar{Y}_1\bar{Y}_2\right).
		\label{power-exponent-zero-rates-Decoder2}
		\end{align}
	\end{theorem}

	\begin{proof}
		Achievability follows by specializing Theorem~\ref{theorem-lower-bounds-power-exponents-general-hypotheses-positive-rates} to $R_1=R_2=0$. The form in \eqref{power-exponent-zero-rates-Decoder1} and \eqref{power-exponent-zero-rates-Decoder2} is then obtained through algebraic manipulations and  application of the log-sum inequality.
		\pe{ The converse can be proved by invoking a slight variation of~\cite[Theorem 3]{shalaby1992multiterminal} in which the distributions are trivariate, instead of bivariate.}
	\end{proof}

	Notice that Theorem~\ref{theorem-optimal-power-exponents-zero-rate-compression} is a strong converse, i.e., it holds for any  values of $\epsilon_1, \epsilon_2\in(0,1)$. 
	Moreover,  it can be achieved  with only a single bit of communication from the encoder to the Detectors 1 and 2, and with only a single bit of communication from Detector 1 to Detector 2. It suffices that the encoder and  Detector 1 send a single bit that simply indicates whether their observed sequences $X^n$ and $Y_1^n$ are $\delta$-typical according to the marginal laws $P_X$ and $P_{Y_1}$. Detector 1 decides on the null hypothesis if both these tests are successful, and Detector 2 decides on the null hypothesis if both tests are successful and  also its own observation $Y_2^n$ is $\delta$-typical according to the marginal $P_{Y_2}$. The analysis of this scheme is similar to the analysis in \cite{han_hypothesis_1987}. 
	From this simple coding scheme, one can conclude  that  the optimal error exponent can be attained even if the encoder and both detectors are only told whether their  sequences are typical with respect to the marginals under  $\mc H$;  there is no need for them to observe the exact sequences.

	\section{Concurrent Detection}\label{secIV}
	
	We now turn to the setting in which the two detectors aim at maximizing the  error exponents under  \textit{different} hypotheses, i.e., $i_1 \neq i_2$. Without loss of generality, let 
	\begin{equation}i_1 = 1 \quad \textnormal{and} \quad i_2=2.
	\end{equation}   
\mw{For convenience, assume that for all $(x,y_1,y_2)\in \mathcal{X} \times \mathcal{Y}_1 \times \mathcal{Y}_2$ both probabilities $P^{(1)}(x,y_1)$ and $P^{(2)}(x,y_1,y_2)$ are positive.} 
%are such that $ P^{(1)}_{XY_1}>0$ and $P^{(2)} > 0$. Important in this problem is the number of null hypotheses for which the $X$-marginals are distinct. Let $\{Q^{(1)},\ldots,Q^{(\L)}\}$ denote the set of pmfs in $\{P^{(1)},\ldots,P^{(\M)}\}$ which have distinct $X$-marginals. 

	 %Also, let $\T$ denote the size of the largest set of pmfs in $\{P^{(1)},\ldots,P^{(M)}\}$ with the same $X$-marginal.  The values of $\L$ and $\T$ are thresholds below which the error exponents exhibit some tension between them in the zero-rates regime. When the communication rates are zero but the ranges of the encoding functions are such that $\|\phi_{1,n}\| > \L$ and  $\|\phi_{2,n}\| > \T$ (e.g., $\|\phi_{1,n}\|$ and $\|\phi_{2,n}\|$ grow sub-exponentially with $n$) such a tension disappears. Here we consider only the case where $\{P^{(1)}_{XY_1},\ldots,P^{(\M)}_{XY_1}\}$ are pairwise distinct.
	
%	Recall that in this case we denote $\|\phi_{1,n}\| =W_1$ and $\|\phi_{2,n}\|=W_2$ and  we indicate this through $R_1=0_{\W_1}$ and $R_2=0_{\W_2}$.
	
	% \subsection{Non-concurrent Zero Rate Detection}
	
	%The following proposition provides a full characterization of the set of error exponents that are achievable for the model of Figure~\ref{fig-system-model} in the case $R_1=R_2=0$ with $\|\phi_{1,n}\| > \L$ and $\|\phi_{2,n}\| > \T$. 
	
	\begin{proposition}~\label{proposition: optimal error exponents zero rate compression cooperation case}
	%	If $R_1=0_{\W_1}$ and $R_2=0_{\W_2}$ with $\W_1 > \L$ and $\W_2 > \T$, \pe{ if all $(x,y_1) \in \mc X\times\mc Y_1$ have  positive probabilities under $\mathcal{H}=1$,
	%	$
%		P^{(1)}_{XY_1}(x,y_1) > 0, 
	%	$
	%	if  all $(x,y_1,y_2) \in \mc X\times\mc Y_1\times\mc Y_2$ have  positive probabilities under $\mathcal{H}=2$,
%		$
%		P^{(2)}_{XY_1Y_2}(x,y_1,y_2) > 0, 
%		$}
		For any pair $\epsilon_1,\epsilon_2$, the zero-rates exponents region $\mc E(0,0,\epsilon_1,\epsilon_2)$ is given by the set of all non-negative pairs 
		$(\theta_1,\theta_2)$ satisfying		
		\begin{IEEEeqnarray}{cRl}\label{eq:exponent region for zero rate with enought bits to send answer from 1 to two.}
		\theta_1 &\leq& \min_{m \neq 1} \min_{ \substack{\tilde P_{XY_1}: \\ \tilde P_{X} = P^{(m)}_X, \quad  \tilde P_{Y_1} = P^{(m)}_{Y_1} }  }D( \tilde P_{XY_1} \| 	P^{(1)}_{XY_1})  \label{eq:exponent 1 for zero rate with enought bits to send answer from 1 to two}\\
		\theta_2 &\leq& \min_{m \neq 2}  \min_{ \substack{\tilde P_{XY_1Y_2}: \\ \tilde P_{X} = P^{(m)}_X, \quad  \tilde P_{Y_1} = P^{(m)}_{Y_1},  \\  \tilde P_{Y_2} = P^{(m)}_{Y_2} } } D( \tilde P_{XY_1Y_2} \| 
		P^{(2)}_{XY_1Y_2}). \IEEEeqnarraynumspace \label{eq:exponent 2 1 for zero rate with enought bits to send answer from 1 to two} 
		\end{IEEEeqnarray}
		
	\end{proposition}

	\begin{proof} 
We first propose a coding scheme achieving this performance. Fix $\mu'' > \mu' > \mu >0$. %	The error exponents of Proposition~\ref{proposition: optimal error exponents zero rate compression cooperation case} can be achieved by a simple scheme that compresses each of the observed sequences $x^n$ to a $\W_1$ \textit{long message} and $y^n_1$  to a $\W_2$ \textit{long message}. 
%			The encoder observes $x^n$ and sends $M_1 \in \{1,\ldots,\M+1\}$ such that: 
		\mw{	\begin{IEEEeqnarray*}{rCl} M_1=\left\{ \begin{array}{ll}  
				m, & \textnormal{if }x^n \in \mc T^n_{\mu}(P^{(m)}_X) \textnormal{ for } m\in \{1,\ldots,\M\}, \qquad \\
				\M+1,  & \textnormal{otherwise}. \end{array} \right.
					 \end{IEEEeqnarray*} 
			Given that Detector 1 observes message $M_1=m_1$ and source sequence $Y_1^n=y_1^n$, it does the following. If $m_1=\M+1$, it declares $\hat{\mathcal{H}}_1=1$ and sends $M_2=\M+1$ over the cooperation link to Detector $2$. Otherwise, it checks whether 
			\begin{IEEEeqnarray}{rCl}
				y_1^n\in \mc T^n_{\mu'}(P^{(m_1)}_{Y_1}).
				 \end{IEEEeqnarray} 
				 If successful, Detector 1 declares $\hat{\mathcal{H}}_1=m_1$ and sends $M_2=m_1$. Otherwise, it declares $\hat{\mathcal{H}}_1=1$ and sends $M_2=\M+1$. 
				 
				 Given that Detector 2 observes messages $M_1=m_1$ and $M_2=m_2$ and source sequence $Y_2^n=y_2^n$, it does the following. If $m_2=\M+1$, Detector $2$ declares $\hat{\mathcal{H}}_2=2$. Otherwise, it  checks whether 
				 			\begin{IEEEeqnarray}{rCl}
				 				y_2^n\in \mc T^n_{\mu''}(P^{(m)}_{Y_2}).
		\end{IEEEeqnarray}
		If successful, it declares $\hat{\mathcal{H}}_2=m$. Otherwise it declares $\hat{\mathcal{H}}_2=2$.}
		
To summarize, Detector~1 declares $\hat{\mathcal{H}}_1=m$, for $m\in \pe{\{ 2,\ldots,\M\}}$ if and only if $(x^n,y_1^n) \in \mc T^n_{\mu}(P_X^{(m)}) \times  \mc T^n_{\mu'}(P_{Y_1}^{(m)})$, and  Detector~2 declares $\hat{\mathcal{H}}_2=m$, for $\pe{m\in \{1,\ldots,\M\}}\backslash \{2\}$, if and only if $(x^n,y_1^n,y_2^n) \in \mc T^n_{\mu}(P_X^{(m)}) \times  \mc T^n_{\mu'}(P_{Y_1}^{(m)})\times  \mc T^n_{\mu''}(P_{Y_2}^{(m)})$. The analysis of the scheme is standard and omitted.

		%
		%Since $\text{Pr}\{X^n \in  \mc T^n_{\mu}(P_{X}^{(m)})\} \geq  1- \hat{\lambda}$, $\text{Pr}\{Y^n_1\in  \mc T^n_{\mu'}(P_{Y_1}^{(m)})'\} \geq  1- \hat{\lambda}$ and $\text{Pr}\{Y^n_2 \in  \mc T^n_{\mu'}(P_{Y_2}^{(m)})\} \geq  1- \hat{\lambda}$, the probability of Type-I error $\alpha_{2,n}$ at the Decoder 2 satisfies $\alpha_{2,n} = \text{Pr}\{ X^n Y^n_1 Y^n_2 \in \mc A^c_{2,n,m}\} \leq 3\hat{\lambda}$. Similarly, the probability of Type-I error $\alpha_{1,n}$ at the Decoder 1 satisfies $\alpha_{1,n} = \text{Pr}\{ X^n Y^n_1 \in \mc A^c_{1,n,m}\} \leq 2\hat{\lambda}$.  As $\hat{\lambda} > 0$ is arbitrary, setting $\hat{\lambda}=\min(\epsilon_1/2,\epsilon_2/3)$ for any given $0 < \epsilon_1 < 1 $ and $0 < \epsilon_2 < 1 $ gives $\alpha_{1,n} \leq \epsilon_1$ and $\alpha_{2,n} \leq \epsilon_2$. Furthermore, the probabilities of Type-II error $\beta_{1,n}=\text{Pr}\{{X}^n{Y}^n_1 \in \mc A_{1,n,m}| \mc H= 1  \}$ and $\beta_{2,n}=\text{Pr}\{{X}^n{Y}^n_1{Y}^n_2 \in \mc A_{2,n,m} | \mc H= 2\}$ are evaluated in exactly the same manner as in the proof of Theorem~\ref{theorem-optimal-power-exponents-zero-rate-compression}. 
		
	The converse  can be proved by invoking a slight variation of~\cite[Theorem 3]{shalaby1992multiterminal} in which the distributions are trivariate, instead of bivariate therein.
		%The converse holds by similar arguments as in the converse proof of  Theorem~\ref{theorem-optimal-power-exponents-zero-rate-compression} . %The converse is proved in proposition \ref{appB}.
	\end{proof}
	
\mw{
	\begin{remark}
	In the previously studied scenario where $i_1=i_2$, the optimal exponents region with zero-rate communication can be achieved with single bits of communication. That means, it suffices to send $M_1\in\{1,2\}$ and $M_2\in\{1,2\}$. It can be shown that this is not the case in the scenario considered here, where $i_1 \neq i_2$. Clearly, in above scheme, both $M_1$ and $M_2$ take on value in $\{1,\ldots, \M+1\}$.\footnote{The scheme could easily be changed to have  $M_2\in\{1,2\}$. In fact, in the scheme either $M_2=M_1$ or  $M_2=\M+1$. So, it suffices that Detector 2 sends $1$ to indicate that it agrees with $M_1$ and $2$ to indicate that it disagrees.}
	If $M_1$  is valued in an alphabet of size $\|\phi_{1,n}\|\leq \M$, then the performance in Proposition~\ref{proposition: optimal error exponents zero rate compression cooperation case} is generally not achievable. 
	% If one insists that  $M_1$  takes value in an alphabet of size $\|\phi_{1,n}\|\leq \M$, then the performance in Proposition~\ref{proposition: optimal error exponents zero rate compression cooperation case} is generally not achievable. 
In particular, while the optimal exponents region in Proposition~\ref{proposition: optimal error exponents zero rate compression cooperation case} is a rectangle, this is not true anymore when  $ \| \phi_{1,n}\|\leq \M$. For the special case where $M_2$ is deterministic (i.e., no cooperation is possible), this was already observed in \cite{isit2018}.
	\end{remark}
	}	

% END OF 'REMOVING OUT THEOREM 4'	

	\section{Proofs}~\label{secV}
	%-----------------------------------------------------------------------------
	\subsection{Proof of Theorem~\ref{theorem-lower-bounds-power-exponents-general-hypotheses-positive-rates}}~\label{secV_subsecA}
	%
	%We first describe our codebook construction, the encoding and decisions functions; and then analyze the Type-I and Type-II errors caused by the decision system for both Decoder 1 and Decoder 2. Note that through the random codebook construction, our approach here somehow departs from~\cite{han_hypothesis_1987}, which is of more statistical, rather than information-theoretic, nature -- the analysis of the Type II errors, however, is inspired by~\cite{han_hypothesis_1987} and rely on similar type-counting arguments.
	\textit{1) Preliminaries:} 
	Choose a small positive number $\delta >0$ and a pair of auxiliary random variables $(U, V)$ satisfying the Markov chains
	\begin{IEEEeqnarray}{rCl}
	U \mkv X \mkv (Y_1, Y_2 )
	\\ V \mkv (Y_1,U )\mkv (Y_2,X).
	\end{IEEEeqnarray}
	%We denote by $P_U$ the marginal PMF of the random variable $U$ and by $P_{V|U}$ the conditional PMF of the random variable $V$ given $U$. Also, for this choice of $(U,V)$, 
	Fix the rates
	\begin{align}
	R_1 &= I(U;X) + \xi(\delta) \\
	R_2 &= I(V;Y_1|U) + \xi(\delta),
	\end{align}
	where $\xi(\cdot ) \to 0$ is a function that tends to 0 as its argument tends  to 0. By this choice,  $(U,V)\in \mathcal{S}(R_1,R_2)$.  
	%  $\delta \to 0$ and  $\epsilon'(\delta) \to 0$ as $\delta \to 0$. %Let  $M_U\triangleq 2^{nR_1}$ and $M_V\triangleq 2^{nR_2}$.
	
	\textit{2) Codebook Generation:} 
	Randomly generate the codebook  $\mc C_{U} \triangleq\big\{ u^n(m_1), \: m_1 \in \pe{\{1,\ldots,\MU \}}\big\}$ by drawing each entry of each codeword $u^n(m_1)$ i.i.d. according to  $P_U$. 
	
	For each index $m_1 \in \pe{\{1,\ldots,\MU\}}$, randomly  construct the codebook  $\mc C_{V} (m_1)\triangleq \{v^n(m_2|m_1), \: m_2 \in \pe{\{1,\ldots,\MV\}}\}$ by drawing the $j$-th entry of each codeword $v^n(m_2|m_1)$ according to the conditional pmf $P_{V|U}(\cdot|u_j(m_1))$, where  $u_j(m_1)$ denotes the $j$-th component of codeword $u^n(m_1)$.

	Reveal all  codebooks to all terminals.

	\textit{3) Encoder:} Given that it observes the source sequence $X^n=x^n$, the encoder looks for an index $m_1\in \pe{\{1,\ldots,\MU\}}$ such that 
	\begin{equation}
	(u^n(m_1), x^n) \in \mc T^n_{\delta/8}(P_{UX}).
	\end{equation}
	%(Note that the measure $P_{UX}$ is the one induced by the null hypothesis $\mc H$). 
	
	If no such index $m_1$ is found, the encoder sends the index $m_1=0$ over the common noise-free pipe to both decoders. If one or more indices can be found, the encoder selects one of them uniformly at random and sends it to both decoders.

	\textit{4) Decoder 1:} Given that Decoder 1 receives an index $M_1=m_1$ not equal to 0 and that it observes the source sequence $Y_1^n= y_1^n$, it checks whether 
	\begin{equation}\label{eq:test2}
	(u^n(m_1), y_1^n) \in \mc T^n_{\delta/4}(P_{UY}).
	\end{equation} If the test is successful, Decoder 1 decides on the null hypothesis, i.e., $\hat{\mc H}_1= \bar{\mc H}$.  Otherwise, it decides on the alternative hypothesis $\hat{\mc H}_1=\bar{\mc H}$.
	
	If $m_1 \neq 0$ and \eqref{eq:test2} holds,  Decoder 1 looks for an index $m_2\in\pe{\{1,\ldots,\MV\}}$ such that 
	\begin{equation}(u^n(m_1), v^n(m_2|m_1), y_1^n) \in \mc T^n_{\delta/2}(P_{UVY_1}).
	\end{equation}
	If one or more such indices can be  found, Decoder 1 selects one of them uniformly at random and sends it over the cooperation link to Decoder~2. Otherwise it sends $M_2=0$.

	\textit{5) Decoder 2:} Given that Decoder 2 observes the indices  $M_1=m_1$ and $M_2=m_2$ and the source sequence $Y_2^n=y_2^n$, it checks whether 
	\begin{equation}
	(u^n(m_1), v^n(m_2|m_1), y_2^n) \in \mc T^n_{\delta}(P_{UVY_2}).
	\end{equation} If this check is successful, Decoder~2 decides on the null hypothesis, $\hat{\mc H}_2=\mc H$. Otherwise, it decides on the alternative hypothesis $\hat{\mc H}_2=\bar{\mc H}$. 
	
	\textit{6) Analysis:} The analysis can be performed along similar lines as in \cite{han_hypothesis_1987}. The main difference is the analysis of the  probability of  error under $\mc H = i_2$  at Decoder~2, which is detailed out in the following. %The rest of the analysis is omitted due to space limitations.

	\noindent Define for each $y_2^n \in \mc Y^n_2$ and each pair of indices $(i,j) \in \pe{\{1,\ldots,\MU\}}\times\pe{\{1,\ldots,\MV\}}$ the set
	\begin{equation}
	\mc S_{ij}(y_2^n) := \mc Q_i \times \{ u^n(i)\} \times \{ v^n(j|i)\} \times \mc G_{ij} \times \{ y_2^n\}, \nonumber 
	\end{equation}
	where $\mc Q_i \subseteq \mc X^n$ is the set of all sequences $ x^n$ for which the encoder sends $M_1=i$ to the two decoders, and $\mc G_{ij}\subseteq \mc Y_1^n$ is the set of all $y_1^n$ sequences for which Decoder~1 sends $M_2=j$ over the cooperation link when $M_1=i$.  Note that, by construction the sets $\{Q_i\}$ are disjoint. Also, define 
	\begin{equation}
	\mc J_n := \bigcup_{i=1}^{\MU} \hspace{-1.5pt}\bigcup_{j=1}^{\MV} \bigcup_{ y_2^n\!\colon ( u^n(i),v^n(j|i), y_2^n) \in \mc T^n_{\delta/2}(P_{UVY_2})} \mc S_{ij}( y_2^n). \nonumber 
	\end{equation}
	%
	% More specifically, for a given triple of sequences $(\dv x, \dv y_1, \dv y_2)$, with joint type $X^{(n)}Y^{(n)}_1Y^{(n)}_2 \in \mc P_n(\mc X\times\mc Y_1\times\mc Y_2)$, we have
	%\begin{align}
	%& \text{Pr}\big\{\bar{X}^n\bar{Y}^n_1\bar{Y}^n_2=(\dv x, \dv y_1, \dv y_2)\big\} = \exp\big[-n(H(X^{(n)}Y^{(n)}_1Y^{(n)}_2) \nonumber\\
	%& \qquad  + D(X^{(n)}Y^{(n)}_1Y^{(n)}_2\|\bar{X}\bar{Y}_1\bar{Y}_2))\big].
	%\label{probability-sequences-same-type-under-memoryless-probabilistic-model}
	%\end{align}
	Denote by %$K_{ij}(y_2^n)$ the number of elements $(x^n,  u^n(i),  v^n(j|i),  y_1^n,  y^n_2) \in \mc S_{ij}( y^n_2)$ whose joint type coincides with $X^{(n)}U^{(n)}V^{(n)}Y^{(n)}_1Y^{(n)}_2$. We have~\cite[Lemma 2.3]{csiszar2011information}
	%\begin{equation}
	%K_{ij}(\dv y_2) \leq \exp\big[nH(X^{(n)}Y^{(n)}_1|U^{(n)}V^{(n)}Y^{(n)}_2)\big].
	%\label{bound-cardinality-type-class}
	%\end{equation}
	%\noindent Thus, the number 
	$K(X^{(n)}U^{(n)}V^{(n)}Y^{(n)}_1Y^{(n)}_2)$ the number of all tuples $( x^n,  u^n(i), v^n(j|i),  y^n_1, y^n_2) \in \mc J_n$  that have joint type $X^{(n)}U^{(n)}V^{(n)}Y^{(n)}_1Y^{(n)}_2$. This number  can be bounded as
	\begin{align}
	& K(X^{(n)}U^{(n)}V^{(n)}Y^{(n)}_1Y^{(n)}_2) \nonumber\\
	& \quad \leq \sum_{i=1}^{\MU} \sum_{j=1}^{\MV} \exp\big[nH(X^{(n)}Y^{(n)}_1Y_2^{(n)}|U^{(n)}V^{(n)})\big] \nonumber\\
	%& }|\mc T^n_{[Y_2|UV]\delta''}(\dv u(i)\dv v(i,j))| \nonumber\\
	& \quad \leq \exp\big[n(H(X^{(n)}Y^{(n)}_1Y_2^{(n)}|U^{(n)}V^{(n)}) \nonumber \\
	& \qquad \qquad  \;  +I(U;X)+ I(V;Y_1|U)+ 2 \xi(\delta))\big]. 
	\end{align}
	%swhere $\mu_n(\delta)$ is a function that tends to $0$ as $n \to \infty$.
	Notice also that for a given triple of sequences $(x^n, y_1^n, y_2^n)$ of joint type   $X^{(n)}Y^{(n)}_1Y^{(n)}_2$: \begin{align}
	& \text{Pr}\big[({X}^n,{Y}^n_1,{Y}^n_2)=(x^n,  y_1^n,  y_2^n)\big | \bar{\mc H}\big] \nonumber \\
	& \; = \exp\big[-n(H(X^{(n)}Y^{(n)}_1Y^{(n)}_2)\nonumber \\
	& \qquad \qquad \qquad \quad + D(X^{(n)}Y^{(n)}_1Y^{(n)}_2\|\bar{X}\bar{Y}_1\bar{Y}_2))\big].
	\label{probability-sequences-same-type-under-memoryless-probabilistic-model}
	\end{align} 
	
	Defining \begin{align}
	& k(X^{(n)}U^{(n)}V^{(n)}Y^{(n)}_1Y^{(n)}_2)  \triangleq \nonumber \\ & \quad  H(X^{(n)}Y^{(n)}_1Y^{(n)}_2)  + D(X^{(n)}Y^{(n)}_1Y^{(n)}_2\|\bar{X}\bar{Y}_1\bar{Y}_2)\nonumber \\
	& \quad- H(X^{(n)}Y^{(n)}_1,Y_2^{(n)}|U^{(n)}V^{(n)}Y^{(n)}_2) - I(U;X) \nonumber\\
	&\quad - I(V;Y_1|U), % - H(Y_2|UV).
	\end{align}
	the error probability under $\mc H = i_2$ at Decoder 2 can then be upper bounded as: 
	\begin{IEEEeqnarray}{rCl}
	\beta_{2,n} &\leq& \hspace{-15pt} \sum_{X^{(n)}U^{(n)}V^{(n)}Y^{(n)}_1Y^{(n)}_2}  \hspace{-15pt} K(X^{(n)}U^{(n)}V^{(n)}Y^{(n)}_1Y^{(n)}_2) \nonumber \\
	&& \quad \quad \times \exp\big[-n(H(X^{(n)}Y^{(n)}_1Y^{(n)}_2)\nonumber \\
	&&  \qquad \qquad \qquad \qquad + D(X^{(n)}Y^{(n)}_1Y^{(n)}_2\|\bar{X}\bar{Y}_1\bar{Y}_2))\big]\nonumber \\
	& \leq&  \sum_{X^{(n)}U^{(n)}V^{(n)}Y^{(n)}_1Y^{(n)}_2} \nonumber\\
	&&  \quad \exp\big[-n(k(X^{(n)}U^{(n)}V^{(n)}Y^{(n)}_1Y^{(n)}_2)- 2\xi(\delta)))\big] \nonumber\\*
	\label{proof-bound-type2-error-Decoder2-step1}
	\end{IEEEeqnarray}
	where the sum ranges over all  joint types $X^{(n)}U^{(n)}V^{(n)}Y^{(n)}_1Y^{(n)}_2 \in \mc P^n(\mc X\times\mc U\times\mc V\times \mc Y_1\times\mc Y_2)$ encountered in $\mc J_n$. Since each of these types satisfies the following three inequalities
	%\noindent From the construction of $\mc J_n$, it is clear that if $(\dv x, \dv u,\dv v,\dv y_1,\dv y_2) \in \mc J_n$, then at least $(\dv u, \dv x) \in \mc T^n_{[UX]\delta}$, $(\dv u, \dv v, \dv y_1) \in \mc T^n_{[UVY_1]\delta'}$ and $(\dv u, \dv v, \dv y_2) \in \mc T^n_{[UVY_2]\delta''}$. Thus, the summation in \eqref{proof-bound-type2-error-Decoder2-step1} is only over all joint types that satisfy
	\begin{align}\label{eq:cond1}
	| P_{U^{(n)}X^{(n)}}(u,x) - P_{UX}(u,x)| & \leq \delta/8 \\
	| P_{U^{(n)}V^{(n)}Y^{(n)}_1}(u,v,y_1) - P_{UVY_1}(u,v,y_1)| & \leq \delta/2 \\
	| P_{U^{(n)}V^{(n)}Y^{(n)}_2}(u,v,y_2) - P_{UVY_2}(u,v,y_2)| & \leq \delta \label{eq:cond3}
	\end{align}
	for all $(x,u,v,y_1,y_2) \in \mc X\times\mc U\times\mc V\times\mc Y_1\times\mc Y_2$ and since the number of joint types is upper bounded by $(n+1)^{|\mc U||\mc V||\mc X||\mc Y_1||\mc Y_2|}$, one obtains
	\begin{align}
	&\beta_{2,n} \leq (n+1)^{|\mc U||\mc V||\mc X||\mc Y_1||\mc Y_2|} \nonumber\\
	&\times \max \; %_{X^{(n)}U^{(n)}V^{(n)}Y^{(n)}_1Y^{(n)}_2} 
	\exp\big[-n(k(X^{(n)}U^{(n)}V^{(n)}Y^{(n)}_1Y^{(n)}_2)-2 \xi(\delta))\big], \nonumber 
	% \label{proof-bound-type2-error-Decoder2-step2} 
	\end{align}
	where the maximization is over all types $X^{(n)}U^{(n)}V^{(n)}Y^{(n)}_1Y^{(n)}_2$ satisfying \eqref{eq:cond1}--\eqref{eq:cond3}.
	Taking now the limits $n\to \infty$ and $\delta \to 0$, by  %\eqref{proof-bound-type2-error-Decoder2-step2} and the 
	the continuity of the entropy and relative entropy  and because $\xi(\delta) \to 0$ as $\delta \to 0$, 
	one obtains that the error exponent of the described scheme satisfies \begin{align}
	&\varliminf_{n \to \infty} \frac{1}{n} \log \beta_2 \nonumber \\
	& \quad \geq  \min  \Big[
	D(\tilde{X}\tilde{Y}_1\tilde{Y}_2\|\bar{X}\bar{Y}_1\bar{Y}_2)  + H(\tilde{X}\tilde{Y}_1\tilde{Y}_2)  -
	\nonumber\\
	&\qquad \quad \qquad  H(\tilde{X}\tilde{Y}_1 \tilde{Y}_2|\tilde{U}\tilde{V}) - I({U};{X}) - I(\tilde{V};\tilde{Y}_1|\tilde{U}) \Big], \nonumber 
	%D(P_{\tilde{U}\tilde{V}\tilde{X}\tilde{Y}_1\tilde{Y}_2}\|P_{\bar{U}\bar{V}\bar{X}\bar{Y}_1\bar{Y}_2}),
	\end{align}
	where the minimization is over all joint types $\tilde{U}\tilde{V}\tilde{X}\tilde{Y}_1\tilde{Y}_2 \in \mc L_2(UV)$.
	%
	% \begin{align}
	% D(\tilde{X}\tilde{Y}_1\tilde{Y}_2\|\bar{X}\bar{Y}_1\bar{Y}_2) \nonumber\\
	% & + H(\tilde{X}\tilde{Y}_1\tilde{Y}_2)  - H(\tilde{X}\tilde{Y}_1|\tilde{U}\tilde{V}\tilde{Y}_2) - I(\tilde{U};\tilde{X}) \nonumber\\
	%& - I(\tilde{V};\tilde{Y}_1|\tilde{U}) - H(\tilde{Y}_2|\tilde{U}\tilde{V}) + \mu'_n
	%\label{proof-bound-type2-error-Decoder2-step3}
	%\end{align}
	%where $\tilde{U}\tilde{V}\tilde{X}\tilde{Y}_1\tilde{Y}_2 \in \mc L_2(UV)$ and $\mu'_n \to 0$ as $n \to \infty$.
	Simple algebraic manipulations establish the desired result.
	%\noindent Finally, substituting in ~\eqref{proof-bound-type2-error-Decoder2-step1} using~\eqref{proof-bound-type2-error-Decoder2-step5}, we obtain
	% \begin{align}
	%&\beta_{2,n}(R_1,R_2,\epsilon_1,\epsilon_2) \leq (n+1)^{|\mc U||\mc V||\mc X||\mc Y_1||\mc Y_2|} \nonumber\\
	%&\times \max_{\tilde{U}\tilde{V}\tilde{X}\tilde{Y}_1\tilde{Y}_2 \in \mc L_2(UV)} \exp\big[-n(D(P_{\tilde{U}\tilde{V}\tilde{X}\tilde{Y}_1\tilde{Y}_2}\|P_{\bar{U}\bar{V}\bar{X}\bar{Y}_1\bar{Y}_2})+\mu_n)\big] 
	%\label{proof-bound-type2-error-Decoder2-step6}
	%\end{align}
	%  which implies that the error exponent $\theta_2(R_1,R_2,\epsilon_1,\epsilon_2)$ is lower-bounded as 
	%	\begin{align}
	%	& \theta_2(R_1,R_2,\epsilon_1,\epsilon_2) \geq \liminf_{n \to \infty} \Big(-\frac{1}{n}\log\beta_{2,n}(R_1,R_2,\epsilon_1,\epsilon_2) \Big) \nonumber\\
	%	&\quad \geq \min_{\tilde{U}\tilde{V}\tilde{X}\tilde{Y}_1\tilde{Y}_2 \in \mc L_2(UV)} D(P_{\tilde{U}\tilde{V}\tilde{X}\tilde{Y}_1\tilde{Y}_2}\|P_{\bar{U}\bar{V}\bar{X}\bar{Y}_1\bar{Y}_2})-\mu_n.
	%	\end{align}
	%\noindent As $\mu_n \to 0$ as $n \to \infty$, this completes the proof of Theorem~\ref{theorem-lower-bounds-power-exponents-general-hypotheses-positive-rates}. 
	
	%-----------------------------------------------------------------------------
	\subsection{Proof of the Converse to Theorem~\ref{theorem-lower-bounds-power-exponents-test-against-independence-positive-rates}}~\label{secV_subsecB}
	Fix $\epsilon_1, \epsilon_2\in(0,1)$. Let encoding functions $\phi_{1,n}$ and $\phi_{2,n}$ and  decision functions $\psi_{1,n}$ and $\psi_{2,n}$ be given that satisfy~\eqref{eq-definition-constant-constraints-typeI-errors}  and ~\eqref{eq-definition-achievable-rates}  with $R_2=0$. Let $\alpha_{1,n}$, $\alpha_{2,n}$, $\beta_{1,n}$, and $\beta_{2,n}$ be the error probabilities corresponding to the chosen functions.
	
	%We first prove the following technical lemma, which  will be instrumental in the rest of the proof.
	%\begin{lemma} \label{lemma: bound on conditional divergence}
	%The following holds
	%\begin{align*}
	%D( Y^n_2 M_1  M_2 \|  \bar{Y}^n_2 \bar{M}_1  \bar{M}_2) &= I(Y^n_2 ; M_1 M_2) \nonumber \\
	%& + D(M_1 M_2 \| \bar{M}_1 \bar{M}_2)
	%\end{align*}
	%\end{lemma}
	%
	%\begin{proof}
	%Since $\bar{X}^n$ is independent of $(\bar{Y}_1^n,\bar{Y}_2^n)$ and $\bar{Y}_1^n$ is independent of $\bar{Y}_2^n$, we have $P_{\bar{X}^n \bar{Y}_1^n \bar{Y}_2^n} = P_{X^n}P_{Y_1^n}P_{Y_2^n}$. Thus, $Y_2^n$ is independent of $(X^n, Y_1^n)$. Together with the Markov chain $(\bar{M}_1, \bar{M}_2) \mkv (\bar{X}^n, \bar{Y}_1^n) \mkv \bar{Y}_2^n$, this implies that $(\bar{M}_1, \bar{M}_2)$ is independent of  $\bar{Y}_2^n$. Hence, $P_{\bar{Y}_2^n \bar{M}_1 \bar{M}_2} = P_{Y_2^n} P_{\bar{M}_1 \bar{M}_2}$. Then, we have
	%\begin{align}
	%	\Delta  &:= D( Y^n_2 M_1  M_2 \|  \bar{Y}^n_2 \bar{M}_1  \bar{M}_2) \nonumber \\
	%	%&= D( P_{Y^n_2 M_1  M_2} \|  P_{\bar{Y}^n_2 \bar{M_1}  \bar{M}_2 }) \nonumber \\
	%  & = D( P_{Y^n_2 M_1  M_2} \|  P_{Y^n_2} P_{\bar{M}_1  \bar{M}_2}) \nonumber \\
	%  &\stackrel{(a)}{=}  D( P_{Y^n_2 |  M_1  M_2} \|  P_{Y^n_2}) \nonumber\\
	%	& + D(P _{M_1 M_2}  \| P_{\bar{M}_1 \bar{M}_2}) \nonumber \\
	%		&=  I( Y^n_2 ;  M_1  M_2) + D( M_1 M_2  \| \bar{M}_1 \bar{M}_2) 
	%\end{align}
	%where (a) follows by the chain rule for divergence.
	% \end{proof}
	
	%We now continue with the proof of Theorem~\ref{theorem-lower-bounds-power-exponents-test-against-independence-positive-rates}. It is easy to see that,
	For $i\in\{1,2\}$:
	\begin{align}
	D\big( P_{  \hat{\mc H}_i | \mathcal{H}} || P_{ \hat{\mc H}_i | \bar{\mc H} }\big) &=  - h_2\left(\alpha_{i,n} \right) - \left(1 - \alpha_{i,n} \right) \log{\left(\beta_{i,n} \right)} \nonumber \\
	& \quad - \alpha_{i,n} \log{\left(1-\beta_{i,n} \right) }
	\label{proof-converse-typeII-error-independent-side-information-step1}
	\end{align}
	where $h_2\left(p\right)$ denotes the entropy of a Bernouilli-$(p)$ memoryless source. Since $\alpha_{i,n} \leq \epsilon_i$, for each $i\in\{1,2\}$, Inequality~\eqref{proof-converse-typeII-error-independent-side-information-step1}  yields
	\begin{align}
	\theta_{i,n} &\triangleq - \frac{1}{n}\log{\left(\beta_{i,n} \right)} 
	% & \leq \frac{1}{n}  \frac{1}{1-\epsilon_i} D\left( P_{  \mathcal{\hat{H}}_i | \mathcal{H}} || P_{ \mathcal{\hat{H}}_i | \bar{\mc H} }\right) + \frac{1}{n}\frac{1}{1-\epsilon_i} h_2\left(\alpha_{i,n}\right) \nonumber \\
	\leq \frac{1}{n}  \frac{1}{1-\epsilon_i} D\big( P_{  \mathcal{\hat{H}}_i | \mathcal{H}} || P_{ \mathcal{\hat{H}}_i | \bar{\mc H} }\big) + \mu_{i,n} \nonumber
	% \label{eq: Divergence Bound logarithm of type II error}
	\end{align}
	with  $\mu_{i,n}\triangleq\frac{1}{n}\frac{1}{1-\epsilon_i} h_2\left(\alpha_{i,n}\right)$. Notice that $\mu_{i,n}\rightarrow 0$ as $n \rightarrow \infty$.  \\[0.1ex]
	
	Consider first $\theta_{1,n}$:
	\begin{align}
	\label{eq: Decision I II error}
	\theta_{1,n}  &\leq  \frac{1}{n}  \frac{1}{1-\epsilon_1}  D\big( P_{  \hat{\mc H}_1 | \mathcal{H}} || P_{ \hat{\mc H}_1 | \bar{\mc H} }\big) + \mu_{1,n} \nonumber \\
	&\stackrel{(a)}{\leq}  \frac{1}{n}  \frac{1}{1-\epsilon_1}  D\big( P_{Y^n_1 M_1 | {\mc H}}|| P_{{Y}^n_1{M}_1|{\bar{\mc H}}} \big)  + \mu_{1,n}\nonumber \\
	&\stackrel{(b)}{=}  \frac{1}{n}  \frac{1}{1-\epsilon_1}  I\left(  Y^n_1;M_1  \right)  + \mu_{1,n}\nonumber \\
	&\stackrel{(c)}{=}  \frac{1}{n}  \frac{1}{1-\epsilon_1} \sum_{k=1}^n{H\big(  {Y_1}_k \big) - H\big(  {Y_1}_k  | M_1  {Y_1}^{k-1} \big)}   + \mu_{1,n}\nonumber \\
	&\stackrel{(d)}{\leq}  \frac{1}{n}  \frac{1}{1-\epsilon_1}  \sum_{k=1}^n{H\big(  {Y_1}_k \big) - H\big(  {Y_1}_k  | M_1 {Y_1}^{k-1} {X}^{k-1} \big)} \nonumber \\ 
	& \qquad  + \mu_{1,n}\nonumber \\
	&\stackrel{(e)}{=}  \frac{1}{n}  \frac{1}{1-\epsilon_1}  \sum_{k=1}^n{H\big(  {Y_1}_k \big) - H\big(  {Y_1}_k  |M_1 {X}^{k-1} \big)}  + \mu_{1,n} \nonumber \\
	&\stackrel{(f)}{=}   \frac{1}{n}  \frac{1}{1-\epsilon_1} \sum_{k=1}^n{I\left(  {Y_1}_k  ; U_k \right) }  + \mu_{1,n}\nonumber \\
	&\stackrel{(g)}{=}  \frac{1}{1-\epsilon_1} {I\left(  {Y_1}_Q  ; U_Q | Q\right) }  + \mu_{1,n}\nonumber \\
	&\stackrel{(h)}{=}  \frac{1}{1-\epsilon_1} {I\left(  Y_1(n) ; U(n)    \right) }  + \mu_{1,n} \nonumber 
	\end{align}
	where: $(a)$ follows by the data processing inequality for  relative entropy; $(b)$ holds since $M_1$ and $Y^n_1$ are independent under the alternative hypothesis $\bar{\mc H}$; $(c)$ is due to the chain rule for mutual information; $(d)$ follows since conditioning reduces entropy; $(e)$ is due to the Markov chain ${Y_1}^{k-1} \mkv (M_1, X^{k-1}) \mkv {Y_1}_k $; $(f)$ holds by defining $U_k \triangleq(M_1, {X}^{k-1})$; $(g)$ is obtained by introducing a  random variable $Q$ that is uniform over the set $\left\{ 1,\cdots, n \right\}$ and independent of all previously defined random variables; and $(h)$ holds by defining $U(n) \triangleq(U_Q, Q)$ and $Y_{1}(n) \triangleq Y_{1Q}$.
	
	Similarly, one obtains for $\theta_{2,n}$:
	\begin{align}
	\theta_{2,n} &\stackrel{(i)}{\leq} \frac{1}{n} \frac{1}{1 - \epsilon_2} D\big(  P_{Y^n_2 M_1 M_2 |{\mc H}}||P_{ {Y}^n_2{M}_1{M}_2|\bar{ \mc H}}\big)  + \mu_{2,n} \nonumber \\
	&\stackrel{(j)}{=} \frac{1}{n} \frac{1}{1 - \epsilon_2} \big (I\left(  Y^n_2 ; M_1 M_2 \right)  \nonumber \\
	& \qquad +D(  P_{ M_1 M_2 |{\mc H}}||P_{{M}_1{M}_2|\bar{ \mc H}})  \big ) 
	+ \mu_{2,n}\nonumber \\
	&\stackrel{(k)}{\leq} \frac{1}{n} \frac{1}{1 - \epsilon_2}(I\left(  Y^n_2 ; M_1 \right) + I\left(  Y^n_2 ; M_2 | M_1 \right)) \nonumber \\
	& \hspace{2.7cm}+  D(  P_{Y^n_1 M_1  |{\mc H}}||P_{ {Y}^n_1{M}_1|\bar{ \mc H}})   + \mu_{2,n} \nonumber \\
	&\stackrel{(\ell)}{\leq} \frac{1}{n} \frac{1}{1 - \epsilon_2}\big( I(  Y^n_2; M_1)+  \log\|\phi_{2,n}\|  \nonumber\\
	& \qquad \qquad \qquad \qquad +D(  P_{Y^n_1 M_1  |{\mc H}}||P_{ Y^n_1{M}_1|\bar{ \mc H}})\big)  + \mu_{2,n}\nonumber \\
	&\stackrel{(m)}{=}   \frac{1}{n} \frac{1}{1 - \epsilon_2}\left( I\left(  Y^n_2; M_1 \right) + I\left(  Y_1^n; M_1 \right)  \right)+ \tilde{\mu}_{2,n} \nonumber  \\
	&\stackrel{(o)}{\leq}  \frac{1}{1 - \epsilon_2} \left( I\left(  Y_{2}(n); U(n) \right) +I \left(  Y_{1}(n); U(n) \right) \right)+ \tilde{\mu}_{2,n},
	\end{align}
	where $(i)$ follows by the data processing inequality for relative entropy; $(j)$ holds by the independence of the pair $(M_1,M_2)$ with $Y_2^n$ under the alternative hypothesis $\bar{\mc H}$; $(k)$ by the data processing inequality for relative entropy; $(\ell)$ holds since conditioning reduces entropy; $(o)$ follows by proceeding along the steps $(b)$ to $(h)$ above; and $(m)$ holds by defining $\tilde{\mu}_{2,n}\triangleq \log\|\phi_{2,n}\|/(n(1-\epsilon_2))+ \mu_{2,n}$. Notice that  by the assumption $R_2=0$, the term $1/n\log\|\phi_{2,n}\| \to 0$ as $n \to \infty$. Thus, also $\tilde{\mu}_{2,n} \to 0$ as $n \to \infty$. 
	
	We next lower bound the rate $R_1$:
	\begin{align}
	n R_1 &\geq H\left(M_1\right)  \nonumber \\
	&= H\left(M_1\right) - H\left(M_1|X^n\right) \nonumber \\
	&= I\left(M_1;X^n\right) \nonumber \\
	&= \sum_{k=1}^{n}{I\left(M_1;X_k|X^{k-1}\right)} \nonumber \\
	&= \sum_{k=1}^{n}{I\left(X_k|U_k\right)} \nonumber \\
	&= n I\left(X_Q;U_Q|Q\right) \nonumber \\
	&=  n I\left(U(n) ; X(n)\right)  \nonumber
	\end{align}
	
	For any blocklength $n$, the newly defined random variables $X(n),Y_{1}(n),Y_{2}(n) \sim P_{X Y_1 Y_2}$ and $(U(n) \mkv X(n) \mkv Y_{1}(n)Y_{2}(n))$. Letting now the blocklength $n\to \infty$ and $\delta \to 0$,  the asymptotic exponents 
	\begin{IEEEeqnarray}{rCl}
	\theta_{1}& \triangleq&  \varliminf_{n \to \infty} \theta_{1,n}\\
	\theta_{2}& \triangleq&  \varliminf_{n \to \infty} \theta_{2,n}  
	\end{IEEEeqnarray}
	satisfy
	\begin{align}
	\theta_1 &\leq I\left(U ; Y_1\right)\\
	\theta_2 &\leq I\left(U ; Y_1\right) +  I\left(U ; Y_2\right);
	\end{align}
	for some  $U \in \mathcal{S}\left(R_1\right)$.
	This completes the proof of Theorem~\ref{theorem-lower-bounds-power-exponents-test-against-independence-positive-rates}.

% END REMOVING OUT PROOF OF THEOREM 4 -- WHICH IS INCORRECT

	\bibliography{./references-allerton}{}
	\bibliographystyle{myIEEEtran}
	
\end{document}